\documentclass[letterpaper, 10 pt, conference]{ieeeconf}  

\IEEEoverridecommandlockouts                              
\usepackage{amssymb}
\usepackage{amsmath}
\usepackage{graphics}
\usepackage{graphicx}
\usepackage{epsfig}
\usepackage{epstopdf}
\usepackage{mathrsfs}
\usepackage{cite}
\usepackage[mathscr]{euscript}
\usepackage{xcolor}
\usepackage{hyperref}

\newcommand{\enn}{\mathbb{N}}
\newcommand{\Si}{S_i}
\newcommand{\Sidot}{\dot{S}_i}
\newcommand{\Ii}{I_i}
\newcommand{\Iidot}{\dot{I}_i}
\newcommand{\Ri}{R_i}
\newcommand{\Ridot}{\dot{R}_i}

\newcommand{\xess}{x^{\rm ess}}
\newcommand{\ess}{{\rm ESS}}

\newcommand{\Rinf}{R^\infty}

\newcommand{\Js}{J^{\rm s}}

\newcommand{\hs}[1][0.5]{\hspace{-#1mm}}

\newcommand{\poa}{\ensuremath{{\rm PoA}}}

\newtheorem{theorem}{Theorem}[section]

\newtheorem{proposition}[theorem]{Proposition}
\newtheorem{lemma}[theorem]{Lemma}

\newtheorem{definition}{Definition}

\begin{document}

\title{\LARGE \bf
Individual Altruism Cannot Overcome Congestion Effects in a Global Pandemic Game
}

\author{Philip N. Brown, Brandon Collins, Colton Hill, Gia Barboza, and Lisa Hines
\thanks{This work was supported by the National Science Foundation under Grants \#DEB-2032465 and \#ECCS-2013779.}
\thanks{The authors are with the University of Colorado at Colorado Springs, CO 80918, USA.
	{\tt\small \{philip.brown,bcollin3, chill13,gbarboza,lhines\}@uccs.edu}}}

\maketitle

\begin{abstract}
A key challenge in responding to public health crises such as COVID-19 is the difficulty of predicting the results of feedback interconnections between the disease and society.
As a step towards understanding these interconnections, we pose a simple game-theoretic model of a global pandemic in which individuals can choose where to live, and we investigate the global behavior that may emerge as a result of individuals reacting locally to the competing costs of isolation and infection.
We study the game-theoretic equilibria that emerge from this setup when the population is composed of either selfish or altruistic individuals.
First, we demonstrate that as is typical in these types of games, selfish equilibria are in general not optimal, but that all stable selfish equilibria are within a constant factor of optimal.
Second, there exist infinitely-many stable altruistic equilibria; all but finitely-many of these are worse than the worst selfish equilibrium, and the social cost of altruistic equilibria is unbounded.
Our work is in sharp contrast to recent work in network congestion games in which all altruistic equilibria are socially optimal.
This suggests that a population without central coordination may react very poorly to a pandemic, and that individual altruism could even exacerbate the problem.
\end{abstract}




\section{Introduction}

One of the chief challenges inherent in predicting the spread of an infectious disease such as COVID-19 is the difficulty in predicting the effects of endogenous societal responses to the outbreak~\cite{Gevertz2020}.
As an infectious disease spreads in society, the members of society may be expected to react to the disease in ways which mitigate its spread~\cite{Franco2020a}.
Simultaneously, if members of society practice effective social distancing, mask-wearing, and other mitigation techniques, these actions affect the rate at which the disease spreads.
In effect, both processes (the spread of the disease and the spread of social behavior) can be viewed as negative feedbacks on each other, and can lead to highly complex dynamics which can be difficult to model, and perhaps even more difficult to control and influence.

To focus on one particular aspect of the disease/society feedback interaction, this paper poses a simple game-theoretic model of a pandemic to study the interacting effects of individually-directed isolation and the propagation of infectious disease.
Our paper adopts a nonatomic game formulation, which models a large population of individuals as a continuum of agents, with each of the infinitely-many individuals having an infinitesimal effect on those around her.
We allow each individual to choose to live at one of a large set of \emph{locations}; for simplicity we consider the case that each individual selects a location and remains there for the duration of the pandemic.
The competing effects of disease risk and personal isolation influence individual decisions: 
\begin{enumerate}
\item Selecting a densely-inhabited location confers a high risk of contracting the infectious disease but a low cost of isolation;
\item Selecting a sparsely-inhabited location confers a low risk of contracting the disease, but a high personal cost of isolation.
\end{enumerate}
We view this model as a testbed for the notion that in a pandemic, it is not good to have either too little isolation (potentially leading to rapid disease spread) or too much isolation (leading to economic ruin, psychological problems, etc.~\cite{Mohler2020,VanGelder2020a,Usher2020}). 
Note that our setting contains many synergies with the population game models~\cite{Sandholm2009} and other game-theoretic and networked approaches in epidemiology~\cite{Hota2019a,Hota2020,Vrabac2021}.

In this paper, we initiate a study on the aggregate effects of uncoordinated individual choice in such a setting, considering in particular the evolutionarily stable states (ESS) associated with a game-theoretic formulation of the above concepts.
Our core initial finding is that if all individuals are self-interested, then every equilibrium has only a (relatively) small number of inhabited locations, and that in general the quality of equilibria can vary considerably even within a single instance of the game.
In particular, in every instance of the problem (even those with no disease), it is an equilibrium for self-interested individuals to have a single inhabited location which contains all individuals.
However, selfish equilibria reliably never err too far on the side of \emph{much} isolation; in essence, a society comprised of only selfish individuals acts as though it prefers sickness over loneliness.
Furthermore, we show that the worst-case selfish equilibria can never be more than a constant factor worse than the optimal social allocation --- that is, the price of anarchy is bounded~\cite{Papadimitriou2001}.

Following our initial study on selfish behavior, we investigate the effects of \emph{altruistic} behavior; in this context, we adopt the view of~\cite{Chen2014} that an altruistic individual attempts to minimize her total social cost; that is, her own individual cost plus the external cost which she imposes on those around her.
Intuitively, an altruistic individual is reluctant to live at a densely-populated location because she recognizes that doing so puts her at risk for being a vector of the disease.
First, we show that when all individuals are altruistic, it is never an equilibrium for the entire population to inhabit a single location, provided that the cost of isolation is not too high; that is, altruists successfully avoid density.
However, we then demonstrate a counterintuitive result which shows that altruists may err quite aggressively on the side of isolation, even when they are correctly calculating and responding to the true marginal costs of their actions.
In fact, every instance of this problem (even if there is no disease) has infinitely-many altruistic ESS, all but finitely-many of which are worse than the worst selfish ESS.

The paper's key message is that in this setting, individual myopic altruism (even when uniform throughout society and computed correctly) is not a reliable tool to guide a society's behavior.
Note that this in sharp contrast to altruism results for a similar nonatomic model of transportation networks, where it is known that uniform altruism ensures that \emph{all} equilibria are socially-optimal~\cite{Chen2014}.
Accordingly, this study on uncoordinated altruism may serve to complement current work in socially-networked autonomy, where a simple and popular approach to societal optimization is to embed autonomous altruists in societal systems~\cite{Byk2018,Li2021,Brown2020a}.
To our knowledge, this paper demonstrates the first-known application for which the price of anarchy of altruism is unbounded --- that is, it represents a cautionary tale about the potential significant negative effects of partial altruism.
More generally, these results illustrate the fact that local alignment of incentive is not sufficient to ensure the presence of high-quality equilibria, even in nonatomic game settings.

\section{Model}

To study the effects of uncoordinated social responses to a pandemic, our model couples a standard density-based SIR dynamical epidemiological model with a nonatomic congestion game.

\subsection{Location-Based Epidemic Model}

We consider a simple extension of standard SIR epidemiological models which allows individuals to select their location from a set of countably-many locations, which for convenience we simply represent as the set of natural numbers $\enn$.
For location $i\in\enn$, we write $x_i\geq0$ to denote the fraction of the total user population that is situated at location $i$.
{Without loss of generality, we index the locations such that $x_i\leq x_{i+1}$ for all $i\in\enn$.}
In our model, we assume that each location represents a fixed physical area and that all locations are the same size; thus, $x_i$ represents both the \emph{number} of individuals at that location and the \emph{density} of the population at location $i$.
We write a \emph{social allocation} as $x:=(x_i)_{i\in\enn}$.

In using $x_i$ to represent a density, we adopt the \emph{density-based} compartmental epidemiological model of~\cite{Hu2013}.
Coarsely speaking, this model assumes that the population at a location is uniformly distributed across the area; thus, contact rates scale linearly with density.
At location $i$, $S_i(t)$, $I_i(t)$, and $R_i(t)$ denote the number of susceptible, infected, and recovered individuals as a function of time, which are defined as solutions to the following nonlinear ordinary differential equations:
\begin{align}
\Sidot(t) 	&= -\beta I_i(t)\Si(t) \label{eq:Sdot}\\
\Iidot(t)	&= \beta \Ii(t)\Si(t) - \gamma \Ii(t)  \label{eq:Idot}\\
\Ridot(t) 	&= \gamma \Ii(t). \label{eq:Rdot}
\end{align}
In keeping with standard epidemiological notation, we write $R_0:=\beta/\gamma$ to denote the \emph{basic reproduction number}.

For each location $i$, we assume that a fraction $\eta\in(0,1)$ of the population is initially infected and that the remainder of the population is initially susceptible.
That is, we assume initial conditions of
\begin{align}
\Si(0) 	&= (1-\eta)x_i  \label{eq:S0}\\
\Ii(0)	&= \eta x_i  \label{eq:I0}\\
\Ri(0) 	&= 0. \label{eq:R0}
\end{align}

It is well-known that the solution to the initial value problem defined by~\eqref{eq:Sdot}-\eqref{eq:R0} is unique for any $\eta$, and that for any initial conditions, the infection eventually dies out: $\lim_{t\to\infty}\Ii(t)=0$. 
We write $R_i^\infty(x_i):=\lim_{t\to\infty}R_i(t)$ to represent the number of individuals at location $i$ who eventually contract the virus at some point during the epidemic.%
\footnote{Note that our model assumes perfect mixing among the individuals at each location.} %

It can be shown by standard techniques~\cite{Pakes2014} that for problem~\eqref{eq:Sdot}-\eqref{eq:R0}, $R_i^\infty(x_i)$ satisfies the equation
\begin{equation}
\Ri^\infty(x_i) =  x_i - (1-\eta)x_i \exp({-R_0\Ri^\infty(x_i)}). \label{eq:transcend}
\end{equation}

Our chief interest is to derive the probability with which an individual at location $i$ will contract the virus at some point throughout the epidemic, which we denote by 
\begin{equation}
p(x_i):=  \frac{R_i^\infty(x_i)}{x_i}.
\end{equation}



\subsection{Nonatomic Game}

To study the emergent behavior resulting from uncoordinated social responses to a global pandemic, we define a nonatomic game in which each individual in the overall population can select which location they wish to inhabit for the duration of the pandemic.
In this context, we model the large population as a continuum of infinitely-many infinitesimally-small individuals.

Each individual wishes to avoid becoming infected by the virus, and thus for location $i$ with density $x_i$, an individual views the probability of becoming eventually sickened $p(x_i)$ as a cost.
However, to model the cost of isolation, we assume that it is independently desirable to be located at a populous location, so each individual experiences an additional \emph{isolation cost} 
\begin{equation}
f(x_i) := Cx_i^{-1}
\end{equation}
for some constant $C>0$. 

As we depict in Figure~\ref{fig:plots}, we model the cost experienced by a selfish individual at location $i$ as
\begin{equation} \label{eq:selfcost}
J^{\rm s}_i(x_i) = f(x_i) + p(x_i).
\end{equation}

We model an \emph{allocation} (i.e., an assignment of individuals to locations) with a Lebesgue-measurable function $X:[0,1]\to\enn$ which maps the set of individuals (the unit interval) to the set of locations ($\enn$).
Given an allocation $X$, the population density at location $i$ is given by the measure of individuals selecting $i$.
Formally, the density at $i$ is $x_i:=\mu(X^{-1}(i))$, where $\mu$ is the Lebesgue measure.
We denote the set of \emph{used locations} by $N(x):=\{i\in\enn\ |\ X(a)=i, a\in[0,1]\}$; note that by this definition a location $i\in N(x)$ need not have $x_i>0$ since it may be used by a measure-$0$ set of individuals.
We often write $x$ (a tuple of densities) to denote an allocation; it is to be understood that these densities are consistent with a measurable function $X$ as described above.
The overall social cost $J(x)$ of a particular allocation $x$ is given by the population-weighted sum of all location costs:
\begin{align}
J(x) 	&:= \sum_{i\in N(x)} x_iJ^{\rm s}_i(x_i) \nonumber \\
	&= C|N(x)| + \sum_{i\in N(x)} R_i^\infty(x_i). \label{eq:Jdef}
\end{align}

From this we may derive the subjective cost experienced by an altruistic individual.
Chen and Kempe's work on altruism~\cite{Chen2014} models an altruistic individual as experiencing a personal cost that is equal to the gradient of the social cost so that an altruistic individual's profitable deviation to another strategy can only occur if this deviation reduces the social cost.
However, $J(x)$ is discontinuous and thus the altruism model requires a suitable modification.
To this end, for allocation $x$, we call location $i$ \emph{empty} if $i\notin N(x)$; that is, $\forall a\in[0,1]$, $X(a)\neq i$.
The discontinuities in $J(x)$ occur when the first nonatomic agent deviates from a used location to an empty location (alternatively, when an agent deviates \emph{from} a location and it becomes empty).

Accordingly, we represent the altruistic cost of an empty location as $C + \partial R^\infty(0)/\partial x_i$.
That is, the \emph{first} individual who deviates to an empty location experiences a switching cost of $C$; this models the fact that the set $N(x)$ has increased in size.
In essence, the deviating individual is ``charged'' for the increase in cost due to the expansion of $N(x)$.
Thus, the subjective cost experienced by an altruistic individual is
\begin{equation} \label{eq:altcost}
J^{\rm a}_i(x_i) = \left\{
\begin{array}{ll}
C + \frac{\partial R_i^\infty (0)}{\partial x_i} 	&\mbox{ if }i\notin N(x), \\
\frac{\partial R_i^\infty (x_i)}{\partial x_i} &	\mbox{ if }i\in N(x).
\end{array}
\right.
\end{equation}

To begin to understand the limitations of uncoordinated social responses to a pandemic, we investigate the evolutionarily-stable states associated with the above game in which individuals are either all selfish~\eqref{eq:selfcost} or altruistic~\eqref{eq:altcost}.
We say that $\xess$ is an \emph{evolutionarily-stable state} (ESS) if no individual can profitably decrease her cost by changing her location \emph{and} if the gradient of every inhabited location is negative.
This second condition is a notion of {stability}, as it guarantees that at an ESS, defections by a sufficiently small group of individuals are self-defeating.

\begin{definition}
An allocation $\xess$ is an ESS for individual type $z\in\{{\rm a},{\rm s}\}$ if and only if it satisfies both of the following conditions:
\begin{enumerate}
\item If $\xess_i>0$, $\xess_j>0$, and $k\notin N(\xess)$ for $i,j,k\in\enn$, then the costs of locations $i$ and $j$ are equal and no greater than that of location $k$:
\begin{equation} \label{eq:wardrop}
J^z_i\left(\xess_i\right) = J^z_i\left(\xess_j\right) \leq J^z_k\left(0\right).
\end{equation}
\item If $\xess_i>0$ and $|N(x)|>1$, then
\begin{equation}\label{eq:ess}
\frac{\partial J_i^z\left(\xess_i\right)}{\partial x_i} > 0.
\end{equation}
\end{enumerate}
\end{definition}

Condition (1) above is simply a definition of a Nash equilibrium; however, our model admits many Nash equilibria which fail condition (2) which we do not regard as reasonable since they are unstable under a wide range of population dynamics.
For instance, when there is no disease (i.e., $\eta=0$), for any $C$ it is a selfish Nash equilibrium to have $x_i=0.1$ for $i\in\{1,2,\dots,10\}$.
At such an equilibrium any increase in the population at any location \emph{reduces} the cost of that location, making this a self-propagating deviation.
While our model considers only static equilibria, we wish to ensure that if coupled with suitable population dynamics, our model's equilibria would at least be locally stable.
Figure~\ref{fig:plots} illustrates the behavior of ESS for several parameter values for our model.

With all of the above, we specify an instance of a \emph{pandemic location game} by $G=(R_0,\eta,C)$ for basic reproduction number $R_0>0$, initial infected fraction $\eta\in(0,1]$, and isolation cost $C>0$; given this game instance, we write its altruistic version as $G^{\rm a}$.
We write the set of ESS for game $G$ as $\ess(G)$; we denote the optimal social cost by $J^*(G):=\inf_x J(x)$.

\begin{figure}
\begin{center}
\vspace{2mm}
\includegraphics[width=0.52\textwidth]{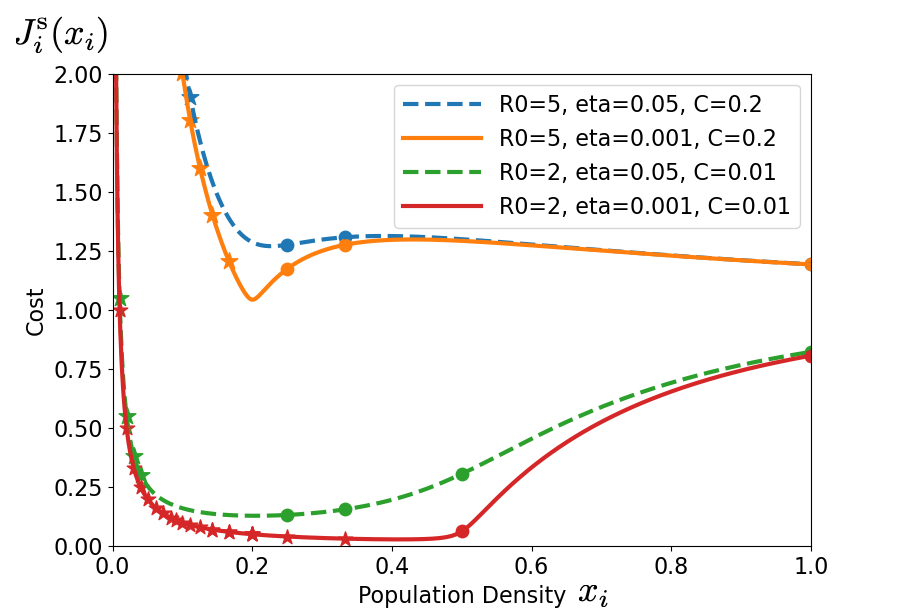}
\end{center}
\vspace{-2mm}
\caption{Plot demonstrating the cost $\Js_i(x_i)$ of a location with respect to population density $x_i$ for various parameter values.
Markers indicate altruistic ($\bigstar$) and selfish ($\bullet$) equilibrium densities.
Note that when initial infected fraction $\eta$ is relatively high (dashed lines), even the best altruistic equilibrium densities can have considerably higher cost than the best selfish equilibrium densities. 
Source code to generate figure is available at
\url{https://github.com/descon-uccs/cdc-2021-altruism-pandemic}
\label{fig:plots}}
\end{figure}
To characterize the relative quality of worst-case ESS , we employ the well-studied price of anarchy~\cite{Papadimitriou2001}:
\begin{equation}
\poa(G) :=  \sup\limits_{x\in\ess(G)}\frac{J(x)}{J^*(G)}
\end{equation}



\section{Our Contributions} \label{sec:ourContributions}

In this paper, we compare the effects of uniform selfishness and uniform altruism in the context of a simple model of a global pandemic.
Conceptually, this may give insight into the extent to which a society may be expected to react to a pandemic in the absence of centralized coordination.
In summary, the key results are that the worst-case cost of all-selfish equilibria is bounded with respect to optimal, but there is no upper bound on the cost of altruistic equilibria.
Concisely, for any pandemic location game $G=(R_0,\eta,C)$, Theorem~\ref{thm:selfishpoa} regarding selfish populations states that
\begin{equation*}
\poa(G)\leq \frac{3}{C} + R_0,
\end{equation*}
but Theorem~\ref{thm:altpoa} regarding altruistic populations states that
\begin{equation*}
\poa(G^{\rm a})  = \infty.
\end{equation*}

We present the full results in a tutorial style; beginning with a useful lemma.


\subsection{A Foundational Lemma}

First we present Lemma~\ref{lem:altruistic} which establishes some properties of $R_i^\infty(x_i)$ which satisfies~\eqref{eq:transcend}.

\begin{lemma} \label{lem:altruistic}
The following are true regarding the partial derivatives of $R_i^\infty(x_i)$ at $x_i=0$ for any problem instance:
\begin{equation}
    \frac{\partial R_i^\infty(0)}{\partial x_i} = \eta,
\end{equation}
and 
\begin{equation}
    \frac{\partial^2 R_i^\infty(0)}{\partial x_i^2} = 2R_0\eta(1-\eta).
\end{equation}
Furthermore, $\frac{\partial^2 R_i^\infty(x_i)}{\partial x_i^2}>0$ on interval $(0,a)$ with $a>0$.
Finally, if $x_i>1/R_0$, then $\frac{\partial R_i^\infty(x_i)}{\partial x_i}>1$ and  if $x_i\leq1/R_0$, then $\frac{\partial R_i^\infty(x_i)}{\partial x_i}\leq1$.
\end{lemma}
\vspace{2mm}

\begin{proof}
Consider~\eqref{eq:transcend}:
\begin{equation*}
\Ri^\infty(x) =  x - (1-\eta)x \exp({-R_0\Ri^\infty(x)}).
\end{equation*}
For compactness, we write $R$, $R'$, and $R''$ to denote $\Ri^{\infty}$ and its first and second partial derivatives with respect to $x$.
By implicit differentiation, we obtain that \begin{equation}
    R' = 1-(1-\eta) \exp(-R_0R)\left[1 -xR_0R'  \right]. \label{eq:rderiv}
\end{equation}
Note that it must be true that $\lim_{x\to0}R_i^\infty(x)=0$, since by definition $R(x)\in[0,x]$ for all $x$; thus, $R$ is right-continuous at $0$ and we may write $R(0)=0$.
It then follows from~\eqref{eq:rderiv} that $R'(0)=\eta$.
Furthermore, note that the above can be rearranged to obtain
\begin{equation}
    R'(1-xR_0(1-\eta)\exp(-R_0R)) = 1-(1-\eta)\exp(-R_0R), \label{eq:Rprime}
\end{equation}
which illustrates that if $x=1/R_0$, then $R'(x)=1$, if $x>1/R_0$ then $R'(x)>1$, and if $x<1/R_0$, then $R'(x)<1$.

Differentiating~\eqref{eq:rderiv} with respect to $x_i$ once again, we obtain that
\begin{equation}
    R'' \hs=\hs -(1-\eta)\exp(-R_0R)\hs\hs\left[ -2R_0R' \hs-\hs
    x\hs\left[R_0R'' \hs-\hs (R_0R')^2
    \right]\right]\hs. \label{eq:rderiv2}
\end{equation}
Once again we may let $R(0)=0$ and $R'(0)=\eta$ to compute that $R''(0)=2R_0\eta(1-\eta)>0$.
Since all derivatives of the above are continuous, it must hold that $R''(x)$ is positive on an interval $(0,a)$ for which $a>0$.
\end{proof}

\subsection{Characterizing Equilibria}

To begin, Propositions~\ref{prop:selfeq} and~\ref{prop:alteq} provide an analytical characterization of the ESS associated with each type of agent.

\begin{proposition}\label{prop:selfeq}
    Let $G$ be a pandemic location game. 
    At any ESS $\xess$ of $G$, it holds that 
    \begin{equation}
    \xess_i=\xess_j \mbox{ for all }i,j\in N(\xess).\label{eq:eqx}
    \end{equation}
    Furthermore,
    \begin{enumerate}
        \item There is a finite number $M_G$ such that if $\xess$ is a selfish ESS, $|N(\xess)|\leq M_G$.
        \item The maximal-density state $x_1=1$ is a selfish ESS.
    \end{enumerate}
\end{proposition}
For intuition, point 1 above says that selfish individuals never err too far on the side of isolation; point 2 says that in fact a selfish population may err aggressively on the side of density (and thus sickness).

\begin{proof}
Equation~\eqref{eq:eqx} holds due to~\eqref{eq:wardrop} and because all cost functions are identical.
To prove item~1 note that $\partial J_i^{\rm s}(x_i)/\partial x_i<0$ on an interval $(0,\bar{x})$ for some $\bar{x}>0$.
Thus, as a consequence of~\eqref{eq:ess} and~\eqref{eq:eqx}, no selfish ESS can have $\xess_i<\bar{x}$ for any $i\in N(\xess)$.
Equivalently, $|N(\xess)|<1/\bar{x}$.
To prove item~2, note that due to the form of $f(x)$, the only possible deviation from the maximal-density state is a deviation to an uninhabited location with an infinite cost.
\end{proof}

Next, we provide an analytical characterization of the altruistic ESS.

\begin{proposition} \label{prop:alteq}
Let $G$ be a pandemic location game.
The following are true regarding the altruistic ESS of $G^{\rm a}$:
\begin{enumerate}
\item At any altruistic ESS $\xess$, if $i\in N(\xess)$, then $J_i^{\rm a}(\xess_i) \leq C + \eta.$ \label{point:alt1}
\item There is no upper bound on the number of locations used in an altruistic ESS.\label{point:alt2}
\item If $C+\eta\leq1$, then the maximal-density allocation $x_1=1$ is never an altruistic ESS. \label{point:alt3}
\end{enumerate}
\end{proposition}

Before presenting the proof, we wish to discuss the ramifications of these points.
Point~\ref{point:alt2} means that altruistic individuals are in some sense ``content'' with high degrees of isolation, even if the isolation is far more aggressive than is required to combat the pandemic.
Point~\ref{point:alt3} shows that in general, altruistic individuals are highly averse to density.
Together, these points imply that altruistic individuals never err too far on the side of density, and that they may in fact err aggressively on the side of isolation.

\begin{proof}
Let $\xess$ be an altruistic ESS.
Due to~\eqref{eq:altcost}, no individual can profitably deviate to an un-used location; that is, for every location $i\in N(\xess)$, it must be true that 
\begin{equation}
    J_i^{\rm a}(\xess_i) \leq C + \frac{\partial R_i^\infty (0)}{\partial x_i}.\nonumber
\end{equation}
Lemma~\ref{lem:altruistic} provides that ${\partial R_i^\infty (0)}/{\partial x_i}=\eta$; thus, if $i\in N(\xess)$, then the fundamental incentive constraint for an altruistic popluation is
\begin{equation}\label{eq:alt incentive constraint}
    J_i^{\rm a}(\xess_i) \leq C + \eta,
\end{equation}
since that is the altruistic cost for an individual to deviate to an empty location.
Thus the proof of Point~\ref{point:alt1} is obtained.

To see Point~\ref{point:alt2}, note that Lemma~\ref{lem:altruistic} provides that $\partial J^{\rm a}(x)/\partial x>0$ on some interval $(0,a)$ with $a>0$.
Thus, if $x$ is such that $x_i=x_j$ for all $i,j\in N(x)$ and $1/|N(x)|<a$, then $x$ is an altruistic ESS.
Note that $x$ may be such that $|N(x)|$ is arbitrarily large and still satisfy this sufficient condition.

Point~\ref{point:alt3} is proved by another application of Lemma~\ref{lem:altruistic}.
This time we appeal to the fact that if $x_i>1/R_0$, then we have $J^{\rm a}_i(x_i)>1$; in particular we have that $J^{\rm a}_i(1)>1$.
If $C+\eta\leq1$, then~\eqref{eq:alt incentive constraint} ensures that $x_1=1$ cannot be an ESS for an altruistic population.
%
\end{proof}

To derive the price anarchy for selfish populations, we prove a pair of lemmas which ultimately provide an upper bound on the cost of any selfish ESS.
First, Lemma~\ref{lem:pprime} shows that for low population densities, the probability of becoming sickened cannot grow too fast as a function of density.

\begin{lemma}\label{lem:pprime}
For every pandemic location game $G=\{R_0,\eta,C\}$, it holds for all $x_i\leq 1/R_0$ that
\begin{equation}
\frac{\partial}{\partial x_i} p_i(x_i) \leq R_0.
\end{equation}
\vspace{2mm}
\end{lemma}
\begin{proof}
Throughout this proof, we omit location index subscripts, writing e.g. $p(x)$ to denote $p_i(x_i)$, and we write $R(x)$ to denote $\Rinf_i(x_i)$.
We write $p'(x)$ and $R'(x)$ to denote the first partial derivatives of $p$ and $R$ with respect to $x$.
First, $p(x)$ is given by
\begin{equation}
p(x) = 1-(1-\eta)\exp(-R_0 R(x)). \label{eq:pimp}
\end{equation}
Performing implicit differentiation, we have
\begin{equation} \label{eq:pprime1}
p'(x) = (1-\eta)R_0R'(x)\exp(-R_0 R(x)).
\end{equation}

Note that~\eqref{eq:pimp} implies that $1-p(x)=(1-\eta)\exp(-R_0R(x))$, so~\eqref{eq:pprime1} can be rewritten as
\begin{align}
p'(x)	&= (1-p(x))R_0R'(x). \nonumber \\
	&\leq R_0,
\end{align}
which holds for all $x\leq1/R_0$, as provided by Lemma~\ref{lem:altruistic}.
\end{proof}

Next, Lemma~\ref{lem:essUB} derives an upper bound on the cost of a selfish ESS that is independent of the initial infected fraction $\eta$.

\begin{lemma} \label{lem:essUB}
For every pandemic location game $G=\{R_0,\eta,C\}$, if $\xess$ is an evolutionarily-stable state for a selfish population, then
\begin{equation}\label{eq:essUB}
J(\xess) \leq \max\{2,CR_0+1\}. 
\end{equation}
\vspace{2mm}
\end{lemma}

\begin{proof}
Let $G=\{R_0,\eta,C\}$ be a pandemic location game and let $\xess$ be an evolutionarily-stable state for a selfish population in $G$; we will write $n:=|N(\xess)|$ to denote the number of inhabited locations in $\xess$.
Proposition~\ref{prop:selfeq} gives that $\xess_i = \xess_j = 1/n$; applying~\eqref{eq:Jdef}, we have for any $i\in\{1,\dots,n\}$ that
\begin{align}
J(x) 	&= Cn + n R_i^\infty(x_i) \nonumber \\
	&= \frac{C}{\xess_i} + \frac{R_i^\infty(\xess_i)}{\xess_i} \nonumber \\
	&= J_i^{\rm s}(\xess_i).
\end{align}
Thus, for simplicity, for the remainder of the proof we will omit $i$ subscripts from functions and write e.g. $\Js(\cdot):=\Js_i(\cdot)$.

First, consider the case that $\xess_i\geq 1/R_0$.
Then
\begin{align}
\Js(\xess_i) 	&= \frac{C}{\xess_i} + p(\xess_i) 	\nonumber \\
			&\leq \frac{C}{\xess_i} + 1 		\nonumber \\
			&\leq CR_0 + 1.
\end{align}

Next, consider the case that $\xess_i< 1/R_0$.
Because $\xess$ is an ESS,~\eqref{eq:ess} gives that the derivative of $\Js$ at $\xess_i$ must be positive, or
\begin{align}
0 	&< \frac{\partial}{\partial \xess_i}\Js(\xess_i) \nonumber \\
	&= -\frac{C}{(\xess_i)^2} + \frac{\partial}{\partial\xess_i}p(\xess_i) \nonumber \\
	&\leq -\frac{C}{(\xess_i)^2}  + R_0, \label{eq:xlb}
\end{align}
Where the last inequality is due to Lemma~\ref{lem:pprime}.
Note that~\eqref{eq:xlb} and preceding can be used to deduce the following lower bound on $\xess_i$:
\begin{equation}
\xess_i > \sqrt{\frac{C}{R_0}}.
\end{equation}
Finally, this may be used to upper-bound $J(\xess)$:
\begin{align}
\Js(\xess_i) 	&< C\sqrt{\frac{R_0}{C}} + p(\xess_i) \nonumber \\
			&\leq \sqrt{CR_0} + 1 \nonumber \\
			&\leq \max\{2, CR_0 + 1\},
\end{align}
completing the proof.
\end{proof}

\addtolength{\textheight}{-6cm}

\subsection{The Worst-Case Cost of Equilibria}

Finally, we are prepared to state our main theorems, showing that selfish populations have a bounded price of anarchy but that altruistic populations do not.

\begin{theorem}\label{thm:selfishpoa}
For every game pandemic location game $G=(R_0,\eta,C)$, the social cost of the worst selfish ESS is bounded above by a constant factor of optimal:
\begin{equation} \label{eq:poa self}
\poa(G) \leq \frac{3}{C} + R_0.
\end{equation}
\end{theorem}
\vspace{2mm}

\begin{proof}
Let $x^*$ denote a global minimizer of $J(x)$, with an associated number of used locations $N^*$.
Note that regardless of the structure of $x^*$, due to~\eqref{eq:Jdef} it satisfies
\begin{equation}
J(x^*)\geq C.
\end{equation}
Applying Lemma~\ref{lem:essUB} yields that
\begin{align}
\poa(G) 	&\leq \frac{\max\{2,1+CR_0\}}{C} \nonumber \\
		&\leq \frac{3}{C} + R_0
\end{align}
as desired, completing the proof.
\end{proof}


\begin{theorem}\label{thm:altpoa}
For every pandemic location game $G$, the social cost of the worst altruistic ESS is unbounded:
\begin{equation}
\poa(G^{\rm a}) = \infty.
\end{equation}
\end{theorem}
\vspace{2mm}

\begin{proof}
Let $G^{\rm a}$ be an altruistic game.
Note that the optimal allocation for $G$ 
\begin{align}
    J^*(G)  &\leq J((x_1=1)) \\
            &\leq C + 1,
\end{align}
since by definition it holds that $R_i^\infty(1)\leq1$.
Now, let $J^{\rm a}(K)$ denote the social cost of an altruistic ESS $\xess$ with $|N(\xess)|=K$.
It must be the case that
\begin{equation}
    J^{\rm a}(K) \geq KC.
\end{equation}
Thus, 
\begin{equation}
    \poa(G^{\rm a}) \geq \frac{KC}{C+1}.
\end{equation}
Due to Proposition~\ref{prop:alteq} (Item~\ref{point:alt2}), $K$ can be selected as large as desired, completing the proof.
\end{proof}

\section{Conclusion}

In the context of a simple game-theoretic model of a pandemic, this paper illustrates that uncoordinated altruism may be as bad as (and possibly considerably worse than) uncoordinated selfishness.
These results hint at a core policy challenge in pandemic management: even individuals who are cognitively capable of assessing their own impact on others (i.e., altruists in our model) may make choices which ultimately lead to inefficiencies, once the choices of others are taken into account.
That is, centralized coordination may be a key requirement of effective pandemic management.





\bibliographystyle{ieeetr} 
\bibliography{library}

\begin{thebibliography}{10}

\bibitem{Gevertz2020}
J.~L. Gevertz, J.~M. Greene, C.~H. Sanchez-Tapia, and E.~D. Sontag, ``{A novel
  COVID-19 epidemiological model with explicit susceptible and asymptomatic
  isolation compartments reveals unexpected consequences of timing social
  distancing},'' {\em Journal of Theoretical Biology}, vol.~510, p.~110539, feb
  2021.

\bibitem{Franco2020a}
E.~Franco, ``{A feedback SIR (fSIR) model highlights advantages and limitations
  of infection-based social distancing},'' {\em
  http://arxiv.org/abs/2004.13216}, 2020.

\bibitem{Mohler2020}
G.~Mohler, A.~L. Bertozzi, J.~Carter, M.~B. Short, D.~Sledge, G.~E. Tita, C.~D.
  Uchida, and P.~J. Brantingham, ``{Impact of social distancing during COVID-19
  pandemic on crime in Los Angeles and Indianapolis},'' {\em Journal of
  Criminal Justice}, vol.~68, no.~March, p.~101692, 2020.

\bibitem{VanGelder2020a}
N.~van Gelder, A.~Peterman, A.~Potts, M.~O'Donnell, K.~Thompson, N.~Shah, and
  S.~Oertelt-Prigione, ``{COVID-19: Reducing the risk of infection might
  increase the risk of intimate partner violence},'' 2020.

\bibitem{Usher2020}
K.~Usher, N.~Bhullar, J.~Durkin, N.~Gyamfi, and D.~Jackson, ``{Family violence
  and COVID‐19: Increased vulnerability and reduced options for support},''
  {\em International Journal of Mental Health Nursing}, 2020.

\bibitem{Sandholm2009}
W.~H. Sandholm, {\em {Population Games and Evolutionary Dynamics}}.
\newblock MIT Press, 2009.

\bibitem{Hota2019a}
A.~R. Hota and S.~Sundaram, ``{Game-Theoretic Vaccination Against Networked SIS
  Epidemics and Impacts of Human Decision-Making},'' {\em IEEE Transactions on
  Control of Network Systems}, vol.~6, pp.~1461--1472, dec 2019.

\bibitem{Hota2020}
A.~R. Hota, T.~Sneh, and K.~Gupta, ``{Impacts of Game-Theoretic Activation on
  Epidemic Spread over Dynamical Networks},'' {\em
  http://arxiv.org/abs/2011.00445}, nov 2020.

\bibitem{Vrabac2021}
D.~Vrabac, M.~Shang, B.~Butler, J.~Pham, R.~Stern, and P.~E. Pare, ``{Capturing
  the Effects of Transportation on the Spread of COVID-19 with a
  Multi-Networked SEIR Model},'' {\em IEEE Control Systems Letters}, pp.~1--1,
  2021.

\bibitem{Papadimitriou2001}
C.~Papadimitriou, ``{Algorithms, games, and the internet},'' in {\em STOC '01:
  Proceedings of the thirty-third annual ACM symposium on Theory of computing},
  pp.~749--753, 2001.

\bibitem{Chen2014}
P.-a. Chen, B.~{De Keijzer}, D.~Kempe, and G.~Shaefer, ``{Altruism and Its
  Impact on the Price of Anarchy},'' {\em ACM Trans. Economics and
  Computation}, vol.~2, no.~4, 2014.

\bibitem{Byk2018}
E.~Bıyık, D.~A. Lazar, R.~Pedarsani, and D.~Sadigh, ``{Altruistic Autonomy:
  Beating Congestion on Shared Roads},'' in {\em International Workshop on the
  Algorithmic Foundations of Robotics}, pp.~887--904, 2018.

\bibitem{Li2021}
R.~Li, P.~N. Brown, and R.~Horowitz, ``{Employing Altruistic Vehicles at
  On-Ramps to Improve the Social Traffic Conditions},'' in {\em 2021 American
  Control Conference (to appear)}, 2021.

\bibitem{Brown2020a}
P.~N. Brown, ``{Partial Altruism is Worse than Complete Selfishness in
  Nonatomic Congestion Games},'' in {\em 2021 American Control Conference (to
  appear)}, 2021.

\bibitem{Hu2013}
H.~Hu, K.~Nigmatulina, and P.~Eckhoff, ``{The scaling of contact rates with
  population density for the infectious disease models},'' {\em Mathematical
  Biosciences}, vol.~244, no.~2, pp.~125--134, 2013.

\bibitem{Pakes2014}
A.~G. Pakes, ``{Lambert's W meets Kermack-McKendrick Epidemics},'' {\em IMA
  Journal of Applied Mathematics (Institute of Mathematics and Its
  Applications)}, vol.~80, no.~5, pp.~1368--1386, 2014.

\end{thebibliography}

\end{document}